\documentclass{article}
\usepackage{spconf,amsmath,graphicx,hyperref}
\usepackage{amssymb}
\usepackage{stfloats}
\usepackage{amsthm}
\usepackage{orcidlink}
\usepackage{pdfpages}
\usepackage{enumitem}
\usepackage{multirow}
\usepackage{amsfonts} 
\usepackage{color}
\usepackage{booktabs}
\usepackage[flushleft]{threeparttable} 
\usepackage{upgreek}
\usepackage{bm}
\usepackage{nicefrac}
\usepackage{cite}
\usepackage{mathtools}
\usepackage{stmaryrd}
\usepackage[mathscr]{euscript}
\usepackage{array}
\usepackage{mathrsfs}
\usepackage{soul}
\newtheorem{thm}{{Theorem}}
\newtheorem*{assumption}{Assumption}

\DeclareMathOperator{\supp}{supp}

 \title{
 Identifiability and Joint Recovery of a Structured Signal Ensemble From Sub-Nyquist Samples}
 
\name{Kumari~Priyanka\orcidlink{0000-0001-7512-1809} and Satish~Mulleti\orcidlink{0000-0002-3995-9070}}
\address{Department of Electrical Engineering,  Indian Institute of Technology Bombay, Mumbai, India, 400076.\\Emails: npriyanka0102@gmail.com, mulleti.satish@gmail.com}

\begin{document}
\ninept
\maketitle

\begin{abstract}
Efficient acquisition of correlated continuous-time signals is critical in applications such as distributed sensor networks and array processing. However, existing correlation models often fail to capture the shared structure of physical signals measured in close proximity. In this paper, we model each signal in an ensemble as the sum of a hidden common lowpass component and a signal-specific innovation highpass component with disjoint spectral supports, where the common bandwidth is unknown. We first derive theoretical identifiability conditions guaranteeing unique decomposition and reconstruction from subsampled observations. Further, we propose a practical recovery algorithm and a joint reconstruction framework that leverages parametric structured dictionaries parameterized by the unknown common bandwidth. Numerical experiments on an ensemble of four signals demonstrate exact reconstruction with a $25\%$ aggregate sampling rate reduction under identifiability conditions, and robust recovery with up to a $74\%$ rate reduction when all channels are sampled below the Nyquist rate, significantly reducing data acquisition and hardware overhead.
\end{abstract}

\begin{keywords}
Correlated bandlimited signals, sub-Nyquist sampling, common-innovation model, structured dictionary learning, low-rate ADCs.
\end{keywords}

\section{Introduction}
In many practical applications, such as telecommunications, radar, and sensor networks, multiple continuous-time signals are acquired simultaneously. Exploiting correlation across observed signals allows for a significant reduction in the aggregate sampling rate. A fundamental model capturing such correlation is the common-innovation structure, where signals in an ensemble share a hidden common component while retaining individual, signal-specific innovation components. This model is well-suited for distributed sensor networks measuring a shared physical field (e.g., temperature or wind velocity) across localized geographical regions, where readings exhibit overall similarity alongside local fluctuations due to factors like elevation or humidity.

Common-innovation structures have been extensively investigated within discrete sparse dictionary frameworks \cite{sarvotham2005distributed, baron2009distributed, duarte2013measurement, oghenekohwo2017low, mahyari2017hierarchical, liu2017efficient, mahyari2017structured, vu2017fast, liu2018common}. In these discrete settings, an ensemble of vectors is represented via common and channel-specific sparsifying dictionaries, enabling joint recovery from compressed measurements. The works considered both known and unknown dictionary frameworks. This framework has also been extended to accommodate unknown time delays in the common component \cite{shiraki2013simultaneous}. In addition, sparse multichannel deconvolution models have also been studied for discrete correlated signals \cite{Hormati2010Distributed, mulleti2020identifiability}.

However, natural signals are intrinsically continuous-time (analog), necessitating sampling strategies that reduce the continuous sampling rate directly. Existing continuous-time approaches typically assume low-rank matrix representations \cite{ahmed2014compressive, ahmed2017compressive, ahmed2019compressive, ahmed2023sub}, relying on random incoherent projections that are often difficult to implement in hardware. Other formulations leverage graph signal processing \cite{ni2022sampling, ravi2024subsampling, sheng2024sampling}, multichannel sparse deconvolution \cite{da2019self,mulleti2021subNyquist}, or stochastic correlation models \cite{mohammadi2017sampling, shlezinger2019joint, jin2025sampling}. Nevertheless, these frameworks do not natively capture the common-innovation structure of continuous-time signals with unknown spectral supports.

In this paper, we consider an ensemble of continuous-time bandlimited signals, each comprising a shared common component and a unique innovation component with disjoint spectral supports. Crucially, the bandwidth of the shared common component is \emph{unknown}, posing significant theoretical and algorithmic challenges. The main contributions of this work are summarized as follows:

\begin{itemize}[leftmargin=*]
    \item \textit{Theoretical Identifiability Guarantees:} We derive sufficient conditions under which an ensemble of $N \geq 3$ continuous-time signals can be uniquely identified and decomposed into their common and innovation components from sub-Nyquist samples, even when the common bandwidth is unknown.
    \item \textit{Parametric Structured Dictionary Framework:} To overcome the sequential nature of the theoretical proof, we develop a practical joint recovery algorithm. We model the common and innovation components using parametric continuous-time dictionaries parameterized by the unknown common bandwidth, formulating recovery as a joint least-squares optimization over finite samples.
    \item \textit{Empirical Validation \& Sub-Nyquist Reduction:} Numerical experiments demonstrate highly accurate recovery when two signals are sampled at Nyquist-rate, and two are sampled sub-Nyquist (reconstruction error ranging from approximately $-300$~dB to $-120$~dB) with a $25\%$ reduction in the aggregate sampling rate. When all channels are sampled sub-Nyquist, experiments show robust recovery across different common bandwidths and ensemble sizes (with average NMSE ranging from approximately $-221$~dB to $-54$~dB) with approximately a $74\%$ aggregate rate reduction.
\end{itemize}

The remainder of the paper is organized as follows. Section~\ref{sec:problem_formulation} presents the signal model and problem formulation. In Section~\ref{sec:identifiability}, we introduce the structural assumptions and establish the main identifiability theorem. Section~\ref{dictionarymethod} details the proposed joint dictionary reconstruction method and numerical evaluations, followed by concluding remarks.

\section{Problem Formulation}
\label{sec:problem_formulation}
Consider an ensemble of $N$ continuous-time correlated bandlimited signals $\{f_n\}_{n=1}^{N}$ with bandwidth $\Omega_m$ such that 
\begin{equation}\label{model}
f_n(t)=f_0(t)+g_n(t), \qquad n=1,\ldots,N, 
\end{equation}
where $f_0(t)$ denotes the common component shared across the signals, which is bandlimited to $[-\Omega_0,\Omega_0]$ with $0<\Omega_0<\Omega_m$. The signal $g_n(t)$ denotes the signal-specific innovation component with spectral support $\mathcal{B}_{\Omega_0, \Omega_m} \coloneq [-\Omega_m, -\Omega_0) \cup (\Omega_0, \Omega_m]$. The Nyquist rate (in rad/s) for each signal is $\Omega_{\text{Nyq}} = 2\Omega_m$.

Signals following the model in \eqref{model} share a common low-frequency (or gross) structure due to $f_0(t)$, with differences arising from the innovation components. The similarity among the signals typically increases as the ratio $\Omega_0/\Omega_m$ approaches one. 

Let $\{f_n(kT_n)\}_{n=1,k\in\mathbb{Z}}^N$ be the observed samples, where $T_n$ denotes the sampling period for the $n$-th signal. The classical Shannon--Nyquist sampling theorem states that a bandlimited signal with bandwidth $\Omega_m$ can be perfectly reconstructed if sampled at a rate greater than or equal to $\Omega_m/\pi$ Hz. Therefore, the required aggregate sampling rate to recover $N$ bandlimited signals without exploiting correlation is at least $\tfrac{N\Omega_m}{\pi}$. However, by leveraging the underlying correlation structure, the aggregate sampling rate can be reduced.

In this paper, we assume that the signals $\{f_n\}_{n=1}^N$ are sampled such that $\sum_{n=1}^N\tfrac{1}{T_n}<\tfrac{N\Omega_m}{\pi}.$ 
Under this framework, the maximum bandwidth $\Omega_m$ and the observations $\{f_n(kT_n):1\leq n\leq N,k\in\mathbb{Z}\}$ are known, whereas the common bandwidth $\Omega_0$, the common component $f_0$, and the innovation components $\{g_n\}_{n=1}^{N}$ are unknown. Our goal is to establish conditions under which these unknown quantities are uniquely identifiable from the observed sub-Nyquist samples and to develop a practical method to estimate $\Omega_0$ and reconstruct the signal ensemble.

%%%%%%%%%%%%%%%%%%%%%%%%%%%%%%%%%%%%%%%%%%%%%%%%%%%%

\section{Structural Assumptions and Identifiability} 
\label{sec:identifiability}
In this section, we lay down the mathematical conditions that ensure the unique identifiability of the signals from their sub-Nyquist samples. Aliasing may occur when a signal is sampled below its Nyquist rate. Since the common component is shared across the signals, it provides the additional information required for recovery. However, the innovation components could also share common components, which might lead to a non-unique decomposition. Therefore, we impose the following structural assumptions on $f_0$ and $\{g_n\}_{n=1}^N$.

\begin{assumption}\label{ass:identifiability}
There exists an unknown $\Omega_0\in(0,\Omega_m)$ such that the following conditions hold:
\begin{enumerate}[
    label=\textnormal{(A\arabic*)},
    ref=\textnormal{(A\arabic*)},
    leftmargin=*]
\item\label{commonsupport}
For every nonempty open interval $I\subset[-\Omega_0,\Omega_0]$, the Fourier transform of $f_0$, denoted as $\mathcal{F}\{f_0\}$, does not vanish almost everywhere on $I$.
\item\label{innovation} Signals indexed by the integer set $\mathcal{J}_N \subset \{1, \dots, N\}$ are sampled at the Nyquist rate, that is, for $n \in \mathcal{J}_N$, $T_n = \pi/\Omega_m$, where $|\mathcal{J}_N|\geq 2$. Moreover, for every nonempty open interval $I\subseteq\mathcal{B}_{\Omega_0,\Omega_m}$, there exist distinct $p,q\in\mathcal{J}_N$ such that $\mathcal{F}\{g_p\}$ and $\mathcal{F}\{g_q\}$ are not identical almost everywhere on $I$.

\item\label{innovationalias} 
To prevent aliasing in the innovation components of the sub-Nyquist signals, we assume that for each $n\notin\mathcal{J}_N$, there exists an integer $q\geq2$ such that the sampling rate $F_{s,n} = 1/T_n$ satisfies
\begin{equation*}
\frac{\Omega_m}{\pi q}\leq F_{s,n}\leq\frac{\Omega_0}{\pi(q-1)}.
\end{equation*}
\end{enumerate}
\end{assumption}

We now present the main identifiability result.

\begin{thm}\label{thm}
Consider the ensemble of bandlimited signals as in \eqref{model}. For $N\geq3$, if Assumptions \ref{commonsupport}--\ref{innovationalias} hold, the signals are uniquely identifiable from their measurements $\{f_n(kT_n): k \in \mathbb{Z}, 1\leq n\leq N\}$, where the aggregate sampling rate is sub-Nyquist, that is, $\sum_{n=1}^{N}\frac{1}{T_n}<\frac{N\Omega_m}{\pi}.$
\end{thm}

\begin{proof}
 Assuming $N\geq3$ and $|\mathcal{J}_N|>1,$ we first reconstruct $f_n,$ for $n\in\mathcal{J}_N$ from their samples using the Shannon-Nyquist theorem. Then, from these signals, we determine $\Omega_0$ as follows. For $p, q\in\mathcal{J}_N,$ $p\neq q,$ define
\begin{equation*}
D(\Omega)
=\sum_{\substack{p, q\in\mathcal J_N\\ p<q}}\left|\mathcal{F}\{f_p\}(\Omega)
-\mathcal{F}\{f_q\}(\Omega)
\right|.
\end{equation*}
Since
$\supp\mathcal{F}\{g_p\}(\Omega)\subseteq\mathcal{B}_{\Omega_0, \Omega_m},$ for every $n\in\mathcal{J_N}$
$$\mathcal{F}\{f_p\}(\Omega)=\mathcal{F}\{f_0\}(\Omega), |\Omega|\leq\Omega_0.$$ 
Thus 
$$D(\Omega)=0\quad |\Omega|\leq\Omega_0.$$  
Since $\supp\mathcal{F}\{f_0\}(\Omega)\in[-\Omega_0, \Omega_0],$ for every $n\in\mathcal{J_N}$
$$\mathcal{F}\{f_n\}(\Omega)=\mathcal{F}\{g_n\}(\Omega),~\Omega\in\mathcal{B}_{\Omega_0, \Omega_m}.$$ 
Therefore, $D$ cannot vanish almost everywhere on any nonempty open interval contained in $\mathcal{B}_{\Omega_0,\Omega_m}$ by Assumption \ref{innovation}.
Consequently,
\begin{equation*}
\Omega_0
=\sup\left\{\Omega>0:D(\omega)=0\ \text{ a.e. for }|\omega|\leq\Omega\right\}.
\end{equation*}
Thus, the common component is obtained
from any reconstructed channel $n\in\mathcal J_N$ by retaining its spectrum
over $[-\Omega_0, \Omega_0]$. For every $n\notin\mathcal J_N$, the innovation
samples are then obtained as
$$
g_n(kT_n)
=
f_n(kT_n)-f_0(kT_n),
\qquad k\in\mathbb Z.
$$
By Assumption \ref{innovationalias}, these bandpass samples uniquely
determine $g_n(t)$. Hence,
$$
f_n(t)=f_0(t)+g_n(t),
\qquad 1\leq n\leq N,
$$
which proves the identifiability of the ensemble.
\end{proof}

The proof relies on the fact that when there are three or more signals, the ensemble can be theoretically recovered, provided at least two signals are sampled at the Nyquist rate while the remaining are subsampled. This structure allows us to uniquely identify the hidden common component and its bandwidth, which is then used to reconstruct the sub-Nyquist signals. This strategy is conceptually similar to the sampling methodology in \cite{mulleti2020identifiability, mulleti2021subNyquist} for sparse multichannel blind deconvolution (MBD), where a few channels are measured at the Nyquist rate and the rest at a lower rate.

While Theorem \ref{thm} establishes theoretical guarantees, it does not provide a practical or real-time recovery procedure. To address this, we propose a joint reconstruction framework that estimates the unknown common bandwidth $\Omega_0$ directly from the observed subsamples and simultaneously recovers the signal ensemble.
%%%%%%%%%%%%%%%%%%%%%%%%%%%%%%%%%%%%%%%%%%%%%%%%%%%%%%%%%%%%%%%%%%%%%%%%%%%%%%%%
\section{Dictionary-Based Reconstruction Method}\label{dictionarymethod}
To recover the signals from their low-rate samples, an algorithm in line with the proof of Theorem~\ref{thm} could be employed. However, this approach is impractical. Specifically, it first recovers the signals for $n\in \mathcal{J}_N$, followed by determining $\Omega_0$, and subsequently reconstructing the remaining signals using the estimated $\Omega_0$. Such a sequential process is unsuitable in practice, as all signals are typically observed simultaneously. Moreover, it requires an infinite number of samples and ideal filtering. This motivates the development of a structured dictionary-based method that jointly recovers the signal ensemble.

\subsection{Finite-Dimensional Reconstruction Model}
We now consider the reconstruction of correlated signals following the model in \eqref{model} from a finite number of samples per signal. To this end, we represent the common signal $f_0(t)$ and the innovations $g_n(t)$ as linear combinations of shifted expansion kernels $\phi_{_{\Omega_0}}(t)$ and $\psi_{_{\Omega_0,\Omega_m}}(t)$, respectively. The spectra of these kernels are constant over the frequency intervals $[-\Omega_0, \Omega_0]$ and $\mathcal{B}_{\Omega_0, \Omega_m}$, respectively, and zero elsewhere. The signals are thus modeled as 
\begin{align}
    f_n(t) \approx
\sum_{\ell=0}^{L-1}c_{0,\ell}\, \phi_{_{\Omega_0}}(t-\tau_\ell)
+\sum_{\ell=0}^{L-1}d_{n,\ell}\,\psi_{_{\Omega_0,\Omega_m}}(t-\rho_\ell), \label{eq:dictionary_model}
\end{align}
where $L$ is the model order, $c_{0, \ell}$ and $d_{n, \ell}$ are the coefficients corresponding to the common and innovation dictionaries, respectively, and the atom shifts are given by $\tau_\ell=(\ell-\lfloor L/2\rfloor)\pi/\Omega_0$ and $\rho_\ell=(\ell-\lfloor L/2\rfloor)2\pi/(\Omega_m-\Omega_0)$. Note that in practice, the exact model order may be unknown, and signals may experience unknown time shifts prior to sampling. In this work, we treat the model order $L$ as a fixed design parameter and leave the estimation of unknown time shifts for future work.

Under model \eqref{eq:dictionary_model}, we collect samples $f_n(kT_n)$ for $k \in \mathcal{K}_n$, where $\mathcal{K}_n$ is a finite index set and $F_{s, n} = 1/T_n$ is the corresponding sampling rate (in Hz). The objective is to jointly estimate the unknown bandwidth $\Omega_0$, the common coefficients $\{c_{0,\ell}\}$, and the innovation coefficients $\{d_{n,\ell}\}$ from the observed subsamples. We assume that the maximum bandwidth $\Omega_m$ is known and that the model order $L$ is treated as a design parameter fixed prior to reconstruction, such that
\begin{equation}\label{eq:L}
L\leq\min\limits_{1\leq n\leq N}|\mathcal{K}_n| \quad \text{and} \quad (N+1)L\leq\sum\limits_{n=1}^N|\mathcal{K}_n|.
\end{equation}
In numerical experiments, $L$ is selected via preliminary empirical evaluation and held fixed. Given $L$, we construct the structured dictionaries and unknown coefficient vectors to formulate the optimization problem for estimating $\Omega_0$.

\subsection{Dictionaries and Unknown Coefficients}
Let $\bm f_n=[f_n(kT_n)]_{k\in\mathcal K_n}^T \in\mathbb R^{|\mathcal K_n|}$ be the sample vector for the $n$-th signal. Using \eqref{eq:dictionary_model}, the sample vector is decomposed as 
\begin{equation}\label{eq:dictionary_decomposition}
\bm f_n
\approx
\bm{\Phi}_{n} \, \bm c_0
+
\bm{\Psi}_{n}\, \bm d_n,
\end{equation}
where $\bm c_0=[c_{0,0},c_{0,1},\ldots,c_{0,L-1}]^T\in\mathbb R^L$ is the common coefficient vector and $\bm d_n=[d_{n,0},d_{n,1},\ldots,d_{n,L-1}]^T\in\mathbb R^L$ is the innovation coefficient vector of the $n$-th signal. Here, $\bm{\Phi}_{n} \in \mathbb{R}^{|\mathcal{K}_n| \times L}$ and $\bm{\Psi}_{n} \in \mathbb{R}^{|\mathcal{K}_n| \times L}$ are the dictionaries corresponding to the common and innovation components, respectively. These dictionaries are parameterized by the unknown bandwidth $\Omega_0$. The $(k, \ell)$-th entries of $\bm{\Phi}_{n}$ and $\bm{\Psi}_{n}$ are defined as $\phi_{_{\Omega_0}}(kT_n-\tau_\ell)$ and $\psi_{_{\Omega_0,\Omega_m}}(kT_n-\rho_\ell)$, respectively. Note that, unlike conventional dictionary representations used in \cite{sarvotham2005distributed, baron2009distributed, duarte2013measurement, mahyari2017hierarchical, liu2017efficient, mahyari2017structured, liu2018common}, the representation in \eqref{eq:dictionary_decomposition} differs in two key aspects: first, the dictionaries in \eqref{eq:dictionary_decomposition} possess a parametric continuous-time structure, and second, the representation coefficients in our model are non-sparse. 

Across all $N$ signals, there are $L$ coefficients for the common component and $N \times L$ coefficients for the innovation components, yielding $(N+1)L$ unknown parameters. Let $K_n=|\mathcal K_n|$ denote the number of observed samples in the $n$-th signal. We construct the joint vectors and matrix:
\begin{equation*}
\bm\theta
=
\begin{bmatrix}
\bm c_0^T &
\bm d_1^T &
\cdots &
\bm d_N^T
\end{bmatrix}^T, 
\quad
\bm y
=
\begin{bmatrix}
\bm f_1^T &
\bm f_2^T &
\cdots &
\bm f_N^T
\end{bmatrix}^T,
\end{equation*}
\begin{equation*}\label{dictionarytogether}
\text{and} \quad \bm A(\Omega)
=
\begin{bmatrix}
\bm\Phi_1 & \bm\Psi_1 & \bm 0 & \cdots & \bm 0\\
\bm\Phi_2 & \bm 0 & \bm\Psi_2 & \cdots & \bm 0\\
\vdots & \vdots & \vdots & \ddots & \vdots\\
\bm\Phi_N & \bm 0 & \bm 0 & \cdots & \bm\Psi_N
\end{bmatrix}.
\end{equation*}
The joint observation model is then expressed as
\begin{equation*}\label{joint_est}
\bm y\approx\bm A(\Omega)\bm\theta.
\end{equation*}

For numerical stability, after constructing $\bm A(\Omega)$ for each candidate bandwidth $\Omega$, its nonzero columns are normalized to form $\bar{\bm A}(\Omega) = \bm A(\Omega)\bm D(\Omega)^{-1}$, where $\bm D(\Omega)$ is a diagonal matrix containing the $\ell_2$-norms of the nonzero columns of $\bm A(\Omega)$. To estimate the common bandwidth, we minimize the least-squares residual over all valid candidate bandwidths:
\begin{equation*}\label{objective}
\widehat{\Omega}_0=\arg\min\limits_{\Omega\in(0, \Omega_m)}\left(\min_{\bar{\bm\theta}}
\left\|
\bm y-\bar{\bm A}(\Omega)\bar{\bm\theta}
\right\|_2^2\right).
\end{equation*}
Since $\bar{\bm\theta}=\bm D(\Omega)\bm\theta$, column normalization preserves both the matrix rank and the minimum least-squares residual. Consequently, the coefficient vector corresponding to the estimated bandwidth $\widehat{\Omega}_0$ is obtained by solving
\begin{equation*}\label{coefficient_problem}
\widehat{\bm\theta}
=
\arg\min_{\bm\theta}
\left\|
\bm y-\bm A(\widehat{\Omega}_0)\bm\theta
\right\|_2^2.
\end{equation*}
This solution $\widehat{\bm\theta}$ is unique provided that the joint dictionary matrix $\bm A(\widehat{\Omega}_0)$ has full column rank $(N+1)L$. If $\bm A(\widehat{\Omega}_0)$ is rank deficient, $\widehat{\bm\theta}$ is non-unique.

\subsection{Numerical Experiments}
In this section, we evaluate the performance of the proposed dictionary-based recovery approach. We consider $N=4$ signals with bandwidth $\Omega_m=2\pi\times1000$. The samples are generated using model \eqref{eq:dictionary_model} over the observation interval $t \in [-0.05, 0.05)$ s. For clarity, Figs. \ref{fig1:identifiability} and \ref{fig2:all sub-Nyquist} are plotted over a smaller time window. For each tested ensemble size $N$, we use the fixed model order $L=25$ and choose the sampling rates such that \eqref{eq:L} is satisfied. To assess reconstruction performance, we define the undersampling factor (UF) as $\alpha_n = \pi / (\Omega_m T_n)$ and compute the normalized mean squared error (NMSE in dB) as
\begin{equation*}
\text{NMSE} = 10\log_{10}
\left(
\frac{
\sum_{n=1}^{N}
\|\widehat{f}_n(t) - f_n(t)\|_{L^2}^{2}}
{
\sum_{n=1}^{N}
\|f_n(t)\|_{L^2}^{2}}
\right),
\end{equation*}
where $\widehat{f}_n(t)$ denotes the estimate of $f_n(t)$, and the $L^2(\mathbb{R})$ norms are evaluated analytically in the frequency domain via Parseval's identity and the closed-form Gram matrix formulation. The reported results are averaged over 500 independent Monte Carlo trials, where signal coefficients are generated randomly in each trial. A trial is deemed successful when $\Omega_0$ is estimated within a margin of $\pi$ rad/s (equivalent to 0.5 Hz). Under this setup, we evaluate two scenarios.

\noindent\textbf{Experiment 1: Under Identifiability Conditions.}
We first evaluate sampling configurations that satisfy the sufficient identifiability conditions in Assumption \ref{ass:identifiability}. Setting $\Omega_0 = 2\pi \times 500$ rad/s, two signals are sampled at the Nyquist rate ($2000$ Hz), while the remaining two are subsampled at $1000$ Hz according to Assumption \ref{innovationalias}. This yields $25\%$ reduction in the aggregate sampling rate. The proposed method achieves an average NMSE of $-302$ dB. A representative reconstruction example is shown in Fig.~\ref{fig1:identifiability}, where region intervals with prominent variations are highlighted. The estimated signals (dashed lines) closely match the ground-truth signals (solid lines), confirming perfect reconstruction under the reduced aggregate sampling rate.

\begin{figure}[t]
    \centering
    \includegraphics[width=1\linewidth]{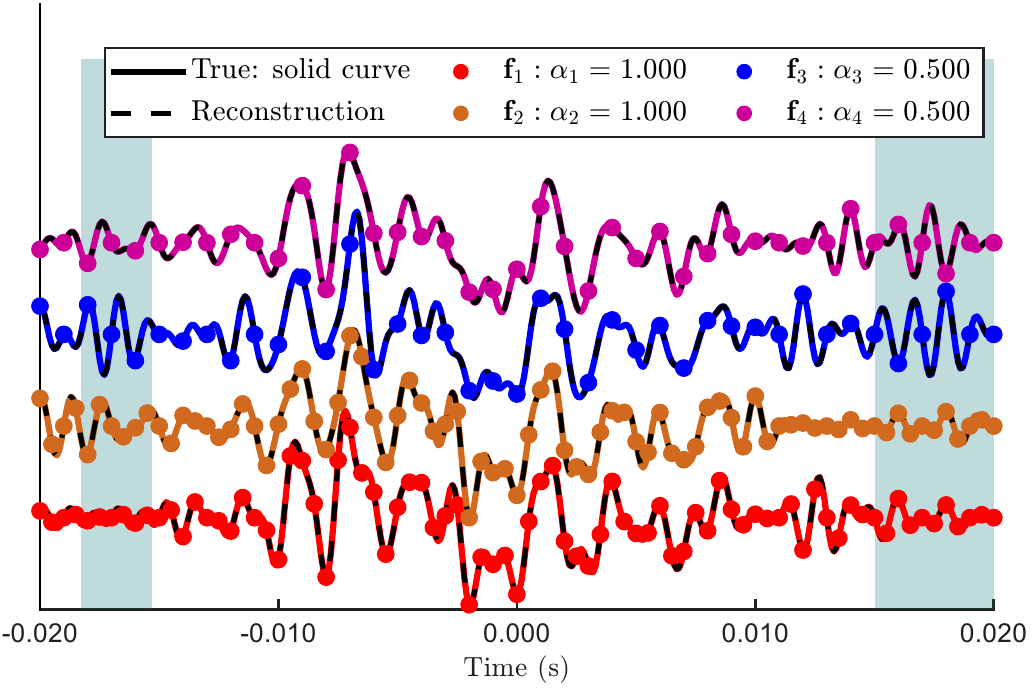}
    \caption{Reconstructed signal ensemble under sufficient identifiability conditions ($\Omega_0=2\pi\times500$ rad/s). Two signals ($f_1, f_2$) are sampled at the Nyquist rate, while two ($f_3, f_4$) are sampled at sub-Nyquist rates according to (A3). The proposed joint recovery method achieves exact reconstruction (NMSE = $-304$ dB) with a $25\%$ aggregate sampling rate reduction. Shaded regions highlight intervals of high inter-signal variance.}
    \label{fig1:identifiability}
\end{figure}

\noindent\textbf{Experiment 2: Sub-Nyquist Sampling Across All Channels.} Next, we reduce the common-to-total bandwidth ratio to $\Omega_0/\Omega_m = 0.2$ (yielding lower inter-signal correlation) and sample all four signals below their Nyquist rates. Specifically, the UFs for the four signals are $\alpha_1= 0.259, \alpha_2= \alpha_3 = 0.261$, and $\alpha_4 = 0.263$, achieving an aggregate sampling rate reduction of $\approx 74\%$. The average NMSE for this setup is $-217$ dB. Figure~\ref{fig2:all sub-Nyquist} illustrates a representative reconstruction, showing accurate signal recovery despite significant structural differences among the signal instances.

Additionally, we evaluated performance across various ensemble sizes $N$ and bandwidth ratios $\tfrac{\Omega_0}{\Omega_m} \in (0, 1)$. Under a fixed aggregate sampling rate reduction, the reconstruction errors remain consistently low (with NMSEs below $-200$~dB), except when the ratio $\tfrac{\Omega_0}{\Omega_m}$ approaches $0$ or $1$ (see Table~\ref{tab:all sub-Nyquist}). Performance degradation in these extreme regimes is driven by ill-conditioning and rank-deficiency of the joint dictionary matrix $\bm A(\Omega_0)$. Specifically, as $\tfrac{\Omega_0}{\Omega_m} \to 0$, the common bandwidth becomes narrow, rendering the common dictionary atoms poorly distinguishable over a finite observation window. Intuitively, this reflects low inter-signal correlation due to a negligible common component. Conversely, as $\tfrac{\Omega_0}{\Omega_m} \to 1$, the bandwidth of the innovation components $(\Omega_m-\Omega_0)$ vanishes, causing the sampled innovation dictionary atoms to become nearly linearly dependent as the signal instances become identical. In both limiting cases, $\bm A(\Omega_0)$ loses column rank, preventing unique recovery of the common and innovation coefficient vectors.

\begin{figure}[t]
    \centering
    \includegraphics[width=1\linewidth]{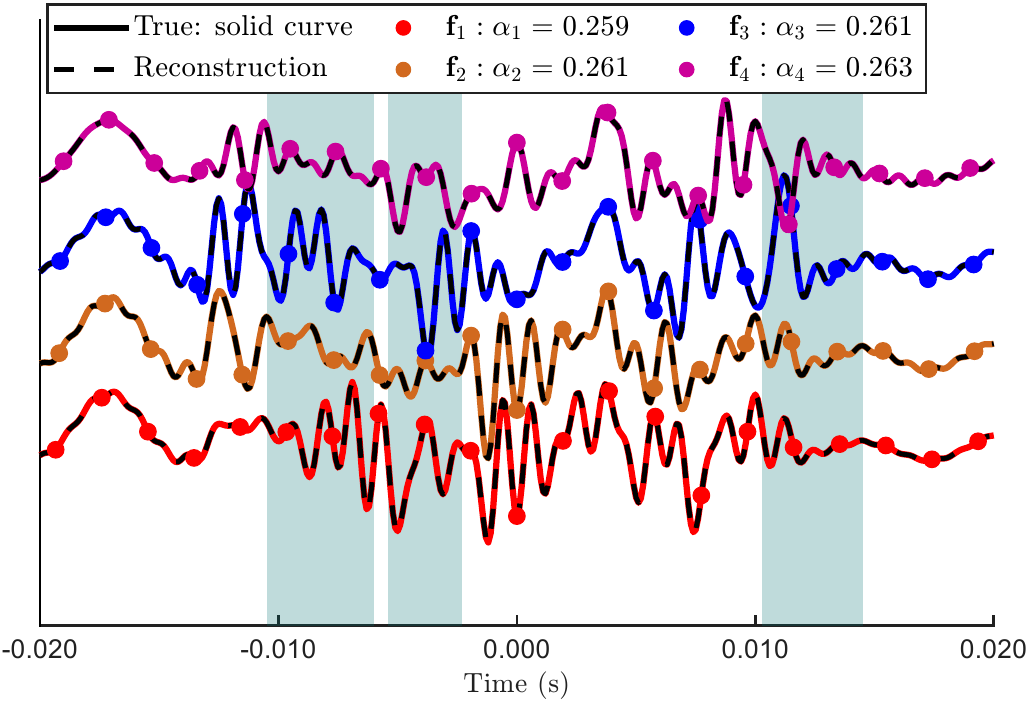}
    \caption{Reconstructed signal ensemble when all channels are sampled below the Nyquist rate ($\alpha_n \approx 0.26$, $\Omega_0=2\pi\times200$ rad/s). Despite low inter-signal correlation and a $74\%$ aggregate sampling rate reduction across all four signals, the parametric dictionary approach achieves robust joint recovery (NMSE = $-221$ dB).}
    \label{fig2:all sub-Nyquist}
\end{figure}
\begin{table}[t]
    \centering
    \caption{Average NMSEs for small and large values of $\Omega_0/\Omega_m$.}
    \begin{tabular}{c|c|c|c}
$N$ & $\Omega_0/\Omega_m$ & Aggregate rate reduction(\%) & Average NMSE\\ 
\hline
5 & 0.04 & $74.6$ &  -54\\
9 & 0.04 & $74.5$  & -64\\
\hline
5 & 0.92 & $74.6$ &  -59\\
9 & 0.92 & $74.5$  &  -62\\
    \end{tabular}
    \label{tab:all sub-Nyquist}
\end{table}

\section{Conclusion}\label{sec:conclusion}
In this paper, we established theoretical identifiability conditions and developed a joint reconstruction algorithm for an ensemble of continuous-time bandlimited signals under a common-innovation correlation model. By capturing the shared lowpass structure and disjoint channel-specific innovation components, our framework enables signal recovery well below the aggregate Nyquist rate, even when the common bandwidth is unknown. To translate theoretical guarantees into a practical simultaneous recovery procedure, we proposed a joint optimization approach using parametric structured dictionaries parameterized by the unknown common bandwidth. Numerical experiments validated our findings, achieving exact recovery with a $25\%$ aggregate sampling rate reduction under sufficient identifiability conditions, and robust recovery with up to a $74\%$ rate reduction when all channels were sampled sub-Nyquist. Future work will explore extensions to handle unknown time shifts in the dictionary atoms, dictionary mismatch, and real-time hardware.
%%%%%%%%%%%%%%%%%%%%%%%%%%%%%%%%%%%%%%%%%%%%%%%%%%%%%%%%%%%%%%%%%%%%%%%%%%%%%%%%%%%%%%%%%%%%%%%%%%%
\bibliographystyle{IEEEbib}
\bibliography{refs}

\end{document}